\documentclass[runningheads,a4paper]{llncs}

\usepackage{ntheorem}
\usepackage{amsmath,amsfonts,url,amssymb, latexsym}
\usepackage[english]{babel}
\usepackage{datetime}
\usepackage{pst-node}

\usepackage{color}

\newcommand{\keywords}[1]{\par\addvspace\baselineskip
\noindent\keywordname\enspace\ignorespaces#1}

\begin{document}

\mainmatter  

\title{Satisfiability of the Two-Variable Fragment of First-Order Logic over Trees}


%
%

\author{Witold Charatonik\inst{1}\fnmsep\thanks{Supported by Polish NCN grant number  DEC-2011/03/B/ST6/00346.} \and
Emanuel Kiero\'nski\inst{1}\fnmsep\thanks{Supported by Polish Ministry of Science and Higher Education grant N N206 371339.} \and  Filip Mazowiecki\inst{2}\fnmsep\footnotemark[1]}

%
\authorrunning{Witold Charatonik, Emanuel Kiero\'nski, Filip Mazowiecki}

\institute{University of Wroc\l aw \and University of Warsaw}

%

\toctitle{Complexity of two-variable logic over finite trees}
\tocauthor{Witold Charatonik, Emanuel Kiero\'nski, Filip Mazowiecki}
\maketitle


\newcommand{\DDD}{\mbox{\large \boldmath $\delta$}}

\newcommand{\bp}{\mathbf{p}}
\newcommand{\br}{\mathbf{r}}
\newcommand{\bP}{\mathbf{P}}
\newcommand{\bR}{\mathbf{R}}
\newcommand{\bE}{\mathbf{E}}
\newcommand{\bL}{\mathbf{L}}
\newcommand{\bO}{\mathbf{O}}
\newcommand{\bU}{\mathbf{U}}

\newcommand{\cA}{\mathcal{A}}
\newcommand{\cB}{\mathcal{B}}
\newcommand{\cC}{\mathcal{C}}
\newcommand{\cD}{\mathcal{D}}%
\newcommand{\cE}{\mathcal{E}}%
\newcommand{\cP}{\mathcal{P}}%
\newcommand{\cQ}{\mathcal{Q}}%
\newcommand{\cK}{\mathcal{K}}%
\newcommand{\cF}{\mathcal{F}}%
\newcommand{\cG}{\mathcal{G}}%
\newcommand{\cL}{\mathcal{L}}%
\newcommand{\cT}{\mathcal{T}}%

\renewcommand{\phi}{\varphi} 

\newcommand{\zz}{\mathit{\!0\!0}}
\newcommand{\zo}{\mathit{\!0\!1}}
\newcommand{\oz}{\mathit{1\!0}}
\newcommand{\oo}{\mathit{1\!1}}

\newcommand{\AAA}{\mbox{\large \boldmath $\alpha$}}
\newcommand{\BBB}{\mbox{\large \boldmath $\beta$}}

\newcommand{\EQ}{\ensuremath{{\mathcal EQ}}}
\newcommand{\Sat}{\ensuremath{\textit{Sat}}}
\newcommand{\FinSat}{\ensuremath{\textit{FinSat}}}

\newcommand{\PDL}{\mbox{\rm PDL}}
\newcommand{\FO}{\mbox{\rm FO}}
\newcommand{\FOt}{\mbox{$\mbox{\rm FO}^2$}}
\newcommand{\FOth}{\mbox{$\mbox{\rm FO}^3$}}

\newcommand{\GF}{\mbox{$\mbox{\rm GF}$}}
\newcommand{\FOtEC}{\mbox{$\mbox{\rm EC}^2$}}
\newcommand{\FOtECth}{\mbox{$\mbox{\rm EC}^2_3$}}
\newcommand{\FOtECt}{\mbox{$\mbox{\rm EC}^2_2$}}
\newcommand{\FOtECo}{\mbox{$\mbox{\rm EC}^2_1$}}
\newcommand{\FOtECk}{\mbox{$\mbox{\rm EC}^2_k$}}
\newcommand{\GFOTC}{\mbox{$\mbox{\rm G}_\exists \mbox{\rm FO}^2+\mbox{\rm TC}$}}
\newcommand{\FOtb}{\mbox{$\mbox{\rm FO}^2[\lessv, \succv]$}}
\newcommand{\FOts}{\mbox{$\mbox{\rm FO}^2[\succv]$}}
\newcommand{\FOto}{\mbox{$\mbox{\rm FO}^2[\lessv]$}}
\newcommand{\GFtb}{\mbox{$\mbox{\rm GF}^2[\lessv, \succv]$}}
\newcommand{\GFts}{\mbox{$\mbox{\rm GF}^2[\succv]$}}
\newcommand{\GFto}{\mbox{$\mbox{\rm GF}^2[\lessv]$}}
\newcommand{\GFt}{\mbox{$\mbox{\rm GF}^2$}}
\newcommand{\FOtEth}{\mbox{$\mbox{\rm EQ}^2_3$}}
\newcommand{\FOtEt}{\mbox{$\mbox{\rm EQ}^2_2$}}
\newcommand{\FOtEo}{\mbox{$\mbox{\rm EQ}^2_1$}}
\newcommand{\FOtEk}{\mbox{$\mbox{\rm EQ}^2_k$}}

\newcommand{\NLogSpace}{\textsc{NLogSpace}}
\newcommand{\NP}{\textsc{NPTime}}
\newcommand{\PTime}{\textsc{PTime}}
\newcommand{\PSpace}{\textsc{PSpace}}
\newcommand{\ExpTime}{\textsc{ExpTime}}
\newcommand{\ExpSpace}{\textsc{ExpSpace}}
\newcommand{\NExpTime}{\textsc{NExpTime}}
\newcommand{\TwoExpTime}{2\textsc{-ExpTime}}
\newcommand{\TwoNExpTime}{2\textsc{-NExpTime}}
\newcommand{\APSpace}{\textsc{APSpace}}
\newcommand{\APTime}{\textsc{APTime}}

\newcommand{\AExpSpace}{\textsc{AExpSpace}}
\newcommand{\ASpace}{\textsc{ASpace}}
\newcommand{\Space}{\textsc{Space}}
\newcommand{\Time}{\textsc{Time}}
\newcommand{\ATime}{\textsc{ATime}}
\newcommand{\AExpTime}{\textsc{AExpTime}}
\newcommand{\DTime}{\textsc{DTime}}
\newcommand{\NPTime}{\textsc{NPTime}}%

\newcommand{\set}[1]{\{#1\}}
\newcommand{\md}[2][] {{\lfloor#2\rfloor_{#1}}}
\newcommand{\sizeOf}[1]{\lVert #1 \rVert}
\newcommand{\str}[1]{{\mathfrak{#1}}}
\newcommand{\restr}{\!\!\restriction\!\!}

\newcommand{\N}{{\mathbb N}}
\newcommand{\Z}{{\mathbb Z}} 

\newcommand{\true}{\mathit{true}}
\newcommand{\false}{\mathit{false}}
\newcommand{\depth}{\mathit{depth}}
\newcommand{\height}{\mathit{height}}
\newcommand{\Root}{\mathit{root}}
\newcommand{\leaf}{\mathit{leaf}}
\newcommand{\elem}{\mathit{elem}}
\newcommand{\transl}[1]{T(#1)}

\newcommand{\taucl}{\tau^{\scriptscriptstyle \#}}

\newcommand{\sss}{\scriptscriptstyle}
\newcommand{\emodels}{\models_{\sss \#}}
\newcommand{\pmodels}{\mid \hspace*{-1.5pt}\approx_{\sss \#}}

\newif\ifdraftpaper
\draftpapertrue

\ifdraftpaper
\newcommand{\nb}[1]{\textcolor{red}{\bf\large \#}\footnote{\textcolor{blue}{#1}}}
\pagestyle{plain}
\else
\newcommand{\nb}[1]{}
\fi

\newcommand{\epr}{\hfill $\Box$}


\newcommand{\succv}{{\downarrow}}
\newcommand{\lessv}{{\downarrow_{\scriptscriptstyle +}}}
\newcommand{\succh}{{\rightarrow}}
\newcommand{\lessh}{{\rightarrow^{\scriptscriptstyle +}}}

\newcommand{\tsuccv}{\theta_{\downarrow}}
\newcommand{\tprecv}{\theta_{\uparrow}}
\newcommand{\tlessv}{\theta_{\downarrow \downarrow_+}}
\newcommand{\tgreatv}{\theta_{\uparrow \uparrow^+}}
\newcommand{\tsucch}{\theta_{\rightarrow}}
\newcommand{\tprech}{\theta_{\leftarrow}}
\newcommand{\tlessh}{\theta_{\rightrightarrows^+}}
\newcommand{\tgreath}{\theta_{\leftleftarrows^+}}
\newcommand{\tfree}{\theta_{\not\sim}}
\newcommand{\teq}{\theta_{=}}

\newcommand{\FOtall}{\mbox{$\FOt[\succv, \lessv, \succh, \lessh]$}}
\newcommand{\cutout}[1]{}

\newcommand{\fw}{{\not\sim}}
\newcommand{\Nn}{\mathbb{N}}
\newcommand{\tree}{\str{T}}
\newcommand{\Ll}{L}
\newcommand{\trs}{I}
\newcommand{\hv}{\hat{v}}
\newcommand{\greatv}{{\uparrow^{\scriptscriptstyle +}}}
\newcommand{\greateqv}{{\uparrow^{\scriptscriptstyle *}}}
\newcommand{\lesseqv}{{\downarrow_{\scriptscriptstyle *}}}

\begin{abstract}
We consider the satisfiability problem for the two-variable fragment of first-order logic over finite unranked trees.
We work with signatures consisting of some unary predicates and the binary navigational predicates $\succv$ (child), $\succh$ (right sibling), and
their respective transitive closures $\lessv$, $\lessh$. We prove that the satisfiability problem for the logic containing all 
these predicates, \FOtall{}, is \ExpSpace-complete.
Further, we consider the restriction of the class of structures to \emph{singular trees}, i.e., we assume that at every node precisely one unary predicate holds. We observe that \FOtall{} and even 
\FOtb{} remain \ExpSpace-complete over finite singular trees, but the complexity decreases for some weaker logics. Namely, the logic with one binary predicate, $\lessv$, denoted \FOto, is \NExpTime-complete, and
its guarded version, \GFto, is \PSpace-complete over finite singular trees, even though both these logics are \ExpSpace-complete over 
arbitrary finite trees. 
\end{abstract}

\section{Introduction}
Classical results from the 1930s by Church and Turing show that the satisfiability problem for first-order logic is undecidable. Moreover, 
undecidability can be proved even for the fragment with only three variables, $\FOth$, \cite{KMW62}. This fact attracted the attention of researchers to the
two-variable fragment, \FOt{}, which turns out to be decidable \cite{Mor75} and \NExpTime-complete \cite{GKV97}. 
In particular, \FOt{} gained a lot of interest  from computer scientists, because of its close connections to formalisms such as modal, temporal, description logics, and XML, widely used in various areas of computer science, including hardware and software verification, knowledge representation, databases, and  artificial intelligence. 

The expressive power of \FOt{} is limited and is not sufficient to axiomatise some natural simple classes of structures, such us trees or words. It is also not possible to say, e.g., that a binary relation is transitive, an equivalence or a linear order.  Thus, \FOt{} over various classes of structures, in which certain relational symbols have to be interpreted in
a special way, e.g., as equivalences, has been extensively studied (see, e.g., \cite{GO98,GradelOR99,Otto01,Kie2011,KO12,KMPT12} for some results in this area). 

\FOt{} over words is investigated in \cite{EVW02}. The authors work there with signatures consisting
of some unary predicates and two built-in binary predicates: $succ$ for the successor relation and $<$ for its transitive closure.
The resulting logic, \FOt$[succ, <]$, is shown to have \NExpTime-complete satisfiability problem, both over $\omega$-words and over
finite words. Actually, the lower bound can be shown for \emph{monadic} \FOt{}, i.e., without using the binary relations $succ$ and $<$. The elementary complexity of \FOt{} over words sharply contrasts with the non-elementary complexity of \FOth{} over words which follows from \cite{Sto74}.

In this paper we consider \FOt{} over unranked trees (ordered or unordered), assuming that, beside unary symbols, signatures may include the child relation $\succv$, the right sibling relation $\succh$, and their respective transitive closures
$\lessv$ and $\lessh$. Decidability of the satisfiability problem for \FOt{} over various classes of infinite trees is implied by the celebrated result by Rabin \cite{Rabin69}, 
that the monadic second-order theory of the binary tree is decidable. Over finite trees decidability follows from \cite{Idziak88}. 
However, regarding complexity, the above mentioned results give only non-elementary upper bounds. 
A better upper complexity bound for the richest of the logics we consider, \FOtall, can be obtained by exploring its correspondence to XPath. In \cite{Marx05} it is argued that \FOtall{} is expressively equivalent to a variant of
Core XPath which is shown in \cite{Marx04} to be \ExpTime-complete. As the translation to XPath involves an exponential blowup in the size of formulas, we get this way \TwoExpTime{} upper bound. 
Our first contribution is establishing the precise complexity of the satisfiability
problem for \FOtall{} over finite trees by showing that it is \ExpSpace-complete. 

Worth mentioning here is the work from \cite{BMSS09}, where two-variable logics over unranked, ordered trees with additional equivalence relation
on nodes, denoted $\sim$, is proposed. The purpose of $\sim$ is to model XML \emph{data values}. It is argued that this extension of \FOtall{} is very hard and  its decidability is left as an
open problem. On the positive side, decidability of \FOt$[\succv, \succh, \sim]$ is shown. 

In the context of XML reasoning it is natural to consider also the additional semantic restriction that at a node of a tree precisely one unary predicate holds. 
We call trees meeting this assumption  \emph{singular trees}. 
In \cite{Weis11} an analogous restriction for finite words is considered.\footnote{In that paper a slightly different terminology is used: the term \emph{word} denotes a structure meeting the singularity assumption, and the term \emph{power words} is reserved for  structures that allow for multiple
unary predicates holding at a single position.}  It appears that \FOt$[succ, <]$ over finite singular words remains \NExpTime-complete, but \FOt$[<]$ becomes \NPTime-complete. In this paper we observe
a similar effect  in the case of unordered trees: over singular trees, \FOtb{} remains \ExpSpace-hard, and the complexity of \FOto{} decreases. This time  the complexity
drop is slightly less spectacular, as the problem is \NExpTime-complete.
We observe, however, that for \NExpTime-hardness the ability of speaking about pairs of elements $x,y$ in free position, i.e., such that $y$ is 
neither an ascendant or descendant of $x$, is needed. This is not typical of logics used in computer science, as their atomic constructions usually allow to refer only to pairs of elements that lie on the same path. To capture the former kind of scenario we consider the restriction of \FOt{} to the two-variable guarded fragment, \GFt{}, in which all quantifiers have to be relativised by binary predicates. We observe that the satisfiability problem for \GFto{} over finite singular trees is \PSpace-complete. To complete the picture we show that augmenting \GFto{} with any of the remaining navigational predicates leads to \ExpSpace-hardness over singular trees. 
Thus, we establish the complexity  over finite trees and over finite
singular trees of all logics \GFt[$\tau_{bin}$] and \FOt[$\tau_{bin}$], for $\{ \lessv \} \subseteq \tau_{bin} \subseteq \{\succv, \lessv, \succh, \lessh \}$. 

\section{Preliminaries}

{\bf Trees and logics.}
We work with signatures of the form $\tau=\tau_0 \cup \tau_{bin}$, where 
$\tau_0$ is a set of unary symbols and $\tau_{bin} \subseteq \{ \succv, \lessv, \succh, \lessh \}$.
Over such signatures we consider two fragments of first-order logic: \FOt{}, i.e., the restriction
of first-order logic in which only variables $x$ and $y$ are available, and \GFt{} being the 
intersection of \FOt{} and the \emph{guarded fragment}, \GF{} \cite{ABN98}.
\GF{} is defined as the least set of formulas such that:
(i) every atomic formula belongs to \GF{};
(ii) \GF{} is closed under logical connectives $\neg, \vee,
\wedge, \Rightarrow$;
and (iii) quantifiers are appropriately relativised by atoms, i.e.,
if $\phi ({\mathbf x}, {\mathbf y})$ is a formula of \GF{} and
$\alpha ({\mathbf x}, {\mathbf y})$ is an atomic formula
containing all the free variables of $\phi$, then the formulas
${\boldsymbol \forall} {\mathbf y}(\alpha ({\mathbf x}, {\mathbf y}) \Rightarrow  \phi ({\mathbf x}, {\mathbf y}))$ and
${\boldsymbol \exists}  {\mathbf y}(\alpha ({\mathbf x}, {\mathbf y}) \wedge \phi ({\mathbf x}, {\mathbf
y}))$ belong to \GF{}.
%
Atom $\alpha ({\mathbf x}, {\mathbf y})$ is called a {\em guard}. 
Equalities $x{=}x$ or $x{=}y$ are also allowed as guards. 

For a given formula $\phi$ we denote by $\tau_0(\phi)$ the set of unary symbols
that appear in $\phi$. 
We write \FOt$[\tau_{bin}]$ or \GFt$[\tau_{bin}]$ to denote that the only binary symbols that are allowed in signatures are those from $\tau_{bin}$.
We are interested in finite unranked tree structures, in which the interpretation of symbols from $\tau_{bin}$ is fixed: if available in
the signature, $\succv$ is interpreted as the child relation, $\succh$ as the right sibling relation, and $\lessv$ and $\lessh$
as their respective transitive closures. If at least one of $\succh$, $\lessh$ is interpreted in a tree then we
say that this tree is \emph{ordered}; in the opposite case we say that the tree is \emph{unordered}. 

We use $x \fw y$ to abbreviate the formula stating that $x$ and $y$ are in \emph{free position},
i.e., that they are related by none of the binary predicates available in the signature. E.g., if we consider 
ordered trees over $\tau_{bin}= \{ \succv, \lessv, \succh, \lessh\}$ then $x \fw y$ can be defined as $x {\not=} y \wedge \neg (x \lessv y) \wedge \neg (y \lessv x) \wedge \neg (x \lessh y) \wedge \neg (y \lessh x)$;
for unordered trees over $\tau_{bin}=\{\lessv \}$ it is just $x {\not=} y \wedge \neg (x \lessv y) \wedge \neg (y \lessv x)$. 

Let us call the formulas specifying the relative position of a pair of elements in a tree with respect to binary predicates \emph{order formulas}.
There are ten possible order formulas: 
$x \succv y$, $y \succv x$, $x \lessv y \wedge \neg (x \succv y)$, $y \lessv x \wedge \neg (y \succv x)$, $x \succh y$, $y \succh x$, $x \lessh y \wedge \neg (x \succh y)$, $y \lessh x \wedge \neg (y \succh x)$, $x \fw y$, $x{=}y$. 
They are denoted, respectively, as: $\tsuccv$, $\tprecv$, $\tlessv$, $\tgreatv$, 
$\tsucch$, $\tprech$, $\tlessh$, $\tgreath$, $\tfree$, $\teq$. Let $\Theta$ be the set of these ten formulas. 


A structure over a signature $\tau=\tau_0 \cup \tau_{bin}$ is \emph{singular} if at every element of this
structure precisely one unary predicate from $\tau_0$ holds. We say that a formula $\phi$ is \emph{singularly satisfiable} (over a class of structures $\cC$)
if there exists a singular model of $\phi$ (from  $\cC$).

We use symbol $\str{T}$ (possibly with sub- or superscripts) to denote tree structures. For a given
tree $\str{T}$ we denote by $T$ its universe.
A \emph{tree frame} is a tree over a signature containing no unary predicates.
We say that a formula $\phi$ is {\emph (singularly) satisfiable over a tree frame} $\cT$ if 
 $\str{T} \models \phi$ for some (singular) $\str{T}$ such that $\cT$ is the restriction of  $\str{T}$ to binary symbols.

\medskip\noindent
{\bf Normal form.}
We say that an \FOtall{} formula $\phi$ is in \emph{normal form} if $\phi=\forall xy \chi(x,y) \wedge \bigwedge_{i \in I} \forall x (\lambda_i(x) \Rightarrow \exists y (\eta_i(x,y) \wedge \psi_i(x,y)))$, for some index set $I$, 
where $\chi(x,y)$ is quantifier-free, $\lambda_i(x)$ is an atomic formula $a(x)$ for some unary symbol $a$, $\psi_i(x,y)$ is a boolean combination of unary atomic formulas, and $\eta_i(x,y)$ is an order formula. Please note, that in $\chi$ the equality symbol may be used, e.g., we can enforce that a model contains at most one node satisfying $a$: $\forall xy (a(x) \wedge a(y) \Rightarrow x{=}y)$.
The following lemma can be proved in a standard fashion (cf.~e.g., \cite{KMPT12}).

\begin{lemma} \label{l:normalform}
Let $\phi$ be an \FOtall{} formula over a signature $\tau$ and let $\cT$ be a tree frame. There exists a polynomially computable \FOtall{} normal form formula $\phi'$ over signature $\tau'$ consisting of $\tau$
and some additional unary symbols, such that $\phi$ is  satisfiable over $\cT$  (singularly satisfiable over $\cT$) iff $\phi'$ is satisfiable over $\cT$ (satisfiable over $\cT$ in a model that restricted to $\tau$ is singular).
 
\end{lemma}

Consider a conjunct $\phi_i=\forall x (\lambda_i(x) \Rightarrow \exists y (\eta_i(x,y) \wedge \psi_i(x,y)))$ of a normal form \FOtall{} 
formula $\phi$. 
Let $\str{T} \models \phi$, and let $v \in T$ be an element such that $\str{T} \models \lambda_i[v]$. Then an element $w \in T$ such
that $\str{T} \models \eta_i[v,w] \wedge \psi_i[v,w]$ is called a \emph{witness} for $v$ and $\phi_i$. 
Sometimes, $b$ is called an \emph{upper} witness if $\eta_i(x,y) \models y \lessv x$, a \emph{lower} witness if $\eta_i(x,y) \models x \lessv y$, and
a \emph{free} witness if $\eta_i(x,y) \models x \fw y$.

\medskip\noindent
{\bf Types.}
A (atomic) {\em $1$-type}, over a signature $\tau=\tau_0 \cup \tau_{bin}$, is a subset of
$\tau_0$. 
We often identify a $1$-type $\alpha$ with the formula $\bigwedge_{a \in \alpha} a(x) \wedge \bigwedge_{a \not\in \alpha} \neg a(x)$.
For a given $\tau$-tree $\str{T}$, and $v \in T$, we denote by ${\rm
  tp}^\str{T}(v)$ the $1$-type \emph{realized} by $v$, i.e., the unique
$1$-type $\alpha$ such that $\str{T} \models \alpha[v]$.

A {\em full type} is a function $\bar{\alpha}:\Theta \rightarrow \cP(\tau_0)$, such that
 $\bar{\alpha}(\tprecv)$, $\bar{\alpha}(\tsucch)$, $\bar{\alpha}(\tprech)$ are singletons or empty, $\bar{\alpha}(\teq)$ is a singleton, and
if $\bar{\alpha}(\tprecv)$ (respectively $\bar{\alpha}(\tsuccv)$, $\bar{\alpha}(\tprech)$, $\bar{\alpha}(\tsucch)$) is empty then 
$\bar{\alpha}(\tgreatv)$ (respectively $\bar{\alpha}(\tlessv)$, $\bar{\alpha}(\tgreath)$, $\bar{\alpha}(\tlessh)$) is also empty.
We employ the following convention: for a given full type $\bar{\alpha}$ we denote by $\alpha$ the unique member of $\bar{\alpha}(\teq)$.
For a given $\tau$-tree $\str{T}$, and $v \in T$, we denote by ${\rm ftp}^\str{T}(v)$  the full type \emph{realized} by $v$, 
i.e., the unique full type $\bar{\alpha}$, such that $\alpha$ is the $1$-type of $v$, 
and for all $\theta \in \Theta$ we have that $\bar{\alpha}(\theta)=
\{ {\rm tp}^{\str{T}}(w): \str{T} \models \theta[v,w] \}$.

A {\em reduced full type} is a tuple $(\alpha, A, B, F)$, where $\alpha$ is a $1$-type
and $A, B, F$ are sets of $1$-types. Reduced full types are used to keep information recorded in full types 
in a slightly (lossy) compressed form.
Let ${\rm ftp}^\str{T}(v)=\bar{\alpha}$. By ${\rm rftp}^\str{T}(v)$ we denote the reduced full type \emph{realized} by $v$, 
i.e., the reduced full type $(\alpha, A, B, F)$, such that  $A=\bar{\alpha}(\tprecv) \cup \bar{\alpha}(\tgreatv)$, 
$B=\bar{\alpha}(\tsuccv) \cup \bar{\alpha}(\tlessv)$ and $F=\bar{\alpha}(\tsucch) \cup \bar{\alpha}(\tprech) \cup \bar{\alpha}(\tlessh) \cup \bar{\alpha}(\tgreath) \cup \bar{\alpha}(\tfree)$.
Note that $\alpha$ denotes the $1$-type of $v$, and, informally speaking, $A$ is the set of $1$-types of elements realized \emph{above} $v$, 
$B$ is the set of $1$-types of elements realized \emph{below} $v$, and $F$ is the set of $1$-types of the siblings of $v$ and the elements realized in \emph{free} position to $v$.

Note that the number
of $1$-types is bounded exponentially, and the numbers of full types and reduced full types are bounded
doubly exponentially in the size of the signature.

For a given normal form \FOtall{} formula $\phi$ and a full type $\bar{\alpha}$, we say that $\bar{\alpha}$ is $\phi$-\emph{consistent} if 
an element realizing $\bar{\alpha}$ cannot be a member of a pair violating the universal conjunct $\forall xy \chi(x,y)$ of $\phi$, and has 
all witnesses required by $\phi$. Formally, $\bar{\alpha}$ is $\phi$-consistent if for every $\theta \in \Theta$, and every $\alpha' \in \bar{\alpha}(\theta)$ we have
$\alpha(x) \wedge \alpha'(y) \wedge \theta(x,y) \models \chi(x,y) \wedge \chi(y,x)$, and 
for every conjunct  
$\forall x (\lambda_i(x) \Rightarrow \exists y (\eta_i(x,y) \wedge \psi_i(x,y)))$ of $\phi$, such that $\alpha(x) \models \lambda_i(x)$,
there exists a $1$-type $\alpha' \in \bar{\alpha}(\eta_i)$ such that $\alpha(x), \alpha'(y) \models \psi_i(x,y)$. 
A proof of the following proposition is straightforward.

\begin{proposition} \label{p:consistent_types}
Let $\str{T}$ be a tree and let $\phi$ be a normal form \FOtall{}-formula. Then $\str{T} \models \phi$ iff
every full type realized in $\str{T}$ is $\phi$-consistent. 
\end{proposition}


We say that a full type $\bar{\alpha}$ is \emph{combined} of two full types $\bar{\alpha}_1$ and $\bar{\alpha}_2$
if $\alpha=\alpha_1=\alpha_2$ and 
for each $\theta \in \Theta$ we have $\bar{\alpha}(\theta)=\bar{\alpha}_1(\theta)$ or 
$\bar{\alpha}(\theta)=\bar{\alpha}_2(\theta)$. Also the following fact is immediate.


\begin{proposition} \label{p:combined_types}
Let $\phi$ be a normal form \FOtall{}-formula, and let $\bar{\alpha}$ be a full type combined of two $\phi$-consistent full types $\bar{\alpha}_1, \bar{\alpha}_2$.
Then $\bar{\alpha}$ is $\phi$-consistent.
\end{proposition}

\section{Finite ordered trees}

This section is devoted to a proof of the following theorem.

\begin{theorem} \label{t:finitetrees}
The satisfiability problem for \FOtall{} over finite trees is \ExpSpace-complete.
\end{theorem}

The crucial fact is that every satisfiable formula has a model of exponentially bounded depth and degree.
We prove this in two steps, and present a procedure looking for such small models,
working in alternating exponential time.


\medskip\noindent
{\bf Short paths.} First, let us see how the paths of a model can be shortened.

\begin{lemma} \label{l:surgery}
Let $\phi$ be a normal form \FOtall{} formula, $\str{T}$ its model,
and $v,w \in T$  two nodes of $\str{T}$, such that $\str{T} \models v \lessv w$ and
${\rm rftp}^\str{T}(v)={\rm rftp}^\str{T}(w)$. Then the tree $\str{T}'$, obtained from $\str{T}$
by replacing the subtree rooted at $v$ by the subtree rooted at $w$,
is a model of $\phi$. 
\end{lemma}
\begin{proof}
It can be verified that for every $u \in T'$, if $u {\not=} w$ then ${\rm ftp}^{\str{T}'}(u) = {\rm ftp}^{\str{T}}(u)$, and that 
${\rm ftp}^{\str{T}'}(w)$ is combined of ${\rm ftp}^{\str{T}}(v)$ and 
${\rm ftp}^{\str{T}}(w)$. Thus, by Propositions \ref{p:consistent_types} and \ref{p:combined_types},  all types realized in $\str{T}'$ are $\phi$-consistent, and $\str{T}' \models \phi$ by Proposition \ref{p:consistent_types}. 
\epr
\end{proof}

Using the above lemma we can successively shorten $\succv$-paths in a model of a normal form 
formula $\phi$ obtaining after a finite number of steps a model of $\phi$ in which on every path only distinct reduced full types
are realized.  Even though there are potentially
doubly exponentially many reduced full types it can be shown that such a model has exponentially bounded $\lessv$-paths.

\begin{lemma} \label{l:shortpaths}
Let $\phi$ be a normal form \FOtall{} formula satisfied in a finite tree. Then there exists a tree model of $\phi$
whose every $\succv$-path has length bounded by $3 \cdot (2^{2\cdot |\tau_0(\phi)|})$, exponentially in $|\phi|$. 
\end{lemma}

\begin{proof}
Let $\str{T} \models \phi$ be a tree in which on every $\succv$-path only distinct full types are realized
and let $v_1, v_2, \ldots, v_k$ be a $\succv$-path in $\str{T}$. Observe that the sets 
$A, B, F$ in reduced full types of $v_i$ behave monotonically. More precisely, if $(\alpha_i, A_i, B_i, F_i)$ is
the reduced full type realized by $v_i$, for $1 \le i \le k$, then for $i<j$ we have 
$A_i \subseteq A_j$, $B_i \supseteq B_j$ and $F_i \subseteq F_j$.
Thus along the path each of the sets $A, B, F$ is modified at most $2^{|\tau_0(\phi)|}$ times (since this is the number of possible $1$-types). 
The number of reduced full types with fixed $A, B, F$ is equal to the number of $1$-types, so the length of each path is
bounded as required. 
\epr
\end{proof}

\noindent
{\bf Small degree.} Now we observe that to provide all witnesses for $\forall \exists$ conjuncts of $\phi$
we only need nodes with at most exponential degree.

\begin{lemma} \label{l:narrowtrees}
Let $\phi$ be a normal form \FOtall{} formula and let $\str{T} \models \phi$. Then there exists a model
$\str{T}' \models \phi$ 
in which the number of successors of each node is bounded by $4 \cdot 2^{2\cdot|\tau_0(\phi)|}$.
Moreover 
$\str{T}'$ can be obtained by removing from $\str{T}$ some number of elements (together with the subtrees rooted at them). 
\end{lemma}

\begin{proof}
We show first how to decrease the degree of a single node of $\str{T}$. 
Let $v$ be a node of $\str{T}$ of full type $\bar{\alpha}_v$, and let $U$ be the set of the children of $v$.
For every element $u \in U$ let $\bar{\alpha}_u$ be its full type. 
We are going to mark some important elements of $U$ and then remove all subtrees rooted at unmarked ones producing a model $\str{T}''' \models \phi$.
First, for every $1$-type $\alpha$, if $\alpha$ is realized in $U$ precisely once then mark this realisation; if $\alpha$ is realized more than once then mark the minimal and the maximal (with respect to $\lessh$) realisations of $\alpha$. 
Further, for every $1$-type $\alpha$, let $U_\alpha=\{u \in U \ \rvert \  \alpha \in \bar{\alpha}_u(\tsuccv) \cup \bar{\alpha}_u(\tlessv)$. 
For each $\alpha$ mark $\mbox{min}(2, |U_\alpha|)$ elements of $U_\alpha$. Note that so far we have marked at most $4 \cdot |2^{\tau_0(\phi)}|$ elements of $U$.
Assume that these (listed according to $\lessh$) are: $u_1, \ldots, u_k$. We call them \emph{primarily marked} elements, and denote their set by $U_P$.

Consider the tree $\str{T}''$ obtained from $\str{T}$ by removing the subtrees rooted at elements of $U \setminus U_P$. 
It can be verified that elements from $T'' \setminus U_P$ retain in $\str{T}''$ their full types from $\str{T}$. 
Unfortunately, the $\succh$-connections among the elements of $U_P$ in $\str{T}''$ may be inconsistent with $\phi$.
To fix this problem we mark some additional elements of $U$ (at most exponentially many) between  $u_i$ and $u_{i+1}$, for all $i$.\footnote{Actually, this fragment of the construction
combined with some earlier parts, reproduces the small model theorem for \FOt{} over words.}
For every $i$, consider the $\succh$-chain $C$ of elements of $\str{T}$ between $u_i$ and $u_{i+1}$. If $C$ is empty then $u_{i+1}$ is $\succh$-successor of $u_i$ and
there is nothing to do. Otherwise, let $\alpha$ be the $1$-type of the successor $w$ of $u_i$. Find the maximal (with respect to $\lessh$) element $w'$ of type $\alpha$ in $C$,
and mark it. The elements between $u_i$ and $w'$ will never be marked, so $w'$ will become the $\succh$-successor of $u_i$ in the final model $\str{T}'''$.
Thus, $u_i$ will retain in its full type its $\bar{\alpha}_{u_i}(\tsucch)$ (singleton) set, and, due to our strategy of primarily marking maximal realisations of $1$-types, also its $\bar{\alpha}_{u_i}(\tlessh)$ set. This is not necessarily true for $w'$ and its (singleton) $\bar{\alpha}(\tprech)$ set, and $\bar{\alpha}_{w'}(\tgreath)$ set. However, 
these sets will be equal, respectively, to  $\bar{\alpha}_w(\tprech)$, and $\bar{\alpha}_w(\tgreath)$ sets of $w$, which means that the full type of $w'$ in $\str{T}'''$ will be combined of two full types
(of $w$ and $w'$) from $\str{T}$. We proceed recursively with the $\succh$-chain of elements between $w'$ and $u_{i+1}$. 

Note that the number of elements between $u_i$ and $u_{i+1}$ which are marked during this process is bounded by the number of $1$-types. Thus we mark
in total  at most
$4 \cdot 2^{|\tau_0(\phi)|} \cdot 2^{|\tau_0(\phi)|}$ elements of $U$, as required in the statement of this lemma. Let us denote the set of the marked elements $U_M$. 
We construct $\str{T}'''$ by removing from $\str{T}$ all subtrees rooted at elements of $U \setminus U_M$. 
It can be verified that all elements from $T''' \setminus U_M$ retain their full types from $\str{T}$,
and that the full types of elements from  $U_M$ in $\str{T}'''$ are either retained from $\str{T}$ 
or are combined of pairs of full types in $\str{T}$ of elements
from $U$. By Proposition \ref{p:consistent_types} we have that $\str{T}''' \models \phi$. 

The desired model $\str{T}'$ can be  obtained by applying the described procedure in depth-first manner.
\epr
\end{proof}

\medskip\noindent
{\bf Alternating procedure and complexity.} We are ready to design a procedure checking if a given \FOtall{} formula $\phi$ has a
finite tree model. By Lemma \ref{l:normalform} we may assume that $\phi$ is in normal form. 
By Lemma \ref{l:shortpaths} and Lemma \ref{l:narrowtrees} we may restrict our attention to models in
which the length of each path and the degree of each node are bounded exponentially in $|\phi|$. 
We present an alternating procedure
 working in exponential time. 
This justifies that the problem is in \ExpSpace{} since, by \cite{CKS81}, \ExpSpace=\AExpTime{}. The procedure first guesses the full type
of the root and then guesses the full types of its children, checking if the information recorded in the full types is locally consistent, and if each
full type is $\phi$-consistent. Further, it works in a loop, universally choosing one of the types of the children and proceeding similarly.

\medskip\noindent
{\bf Procedure} {\tt {FO$^2$[$\succv, \lessv, \succh, \lessh$]-sat-test}}\\
{\bf input:} an \FOtall{} normal form formula $\phi$
\begin{itemize}\itemsep0pt

\item let $maxdepth:=3 \cdot |2^{\tau_0(\phi)}|^2$; let $maxdegree:=4 \cdot 2^{2\cdot|\tau_0(\phi)|}$; 
\item let $level:=0$;
\item {\bf guess} a full type $\bar{\alpha}$ such that $\bar{\alpha}(\tprecv)=\bar{\alpha}(\tgreatv)=\bar{\alpha}(\tsucch)=\bar{\alpha}(\tlessh)=
\bar{\alpha}(\tprech)=\bar{\alpha}(\tgreath)=\bar{\alpha}(\tfree)=\emptyset$; 
\item {\bf while} $level < maxdepth$ {\bf do}
\item \hspace*{20pt} if $\bar{\alpha}$ is not $\phi$-consistent then {\bf reject}
\item \hspace*{20pt} if $\bar{\alpha}(\tsuccv) \cup \bar{\alpha}(\tlessv) = \emptyset$ then {\bf accept}
\item \hspace*{20pt} {\bf guess} an integer $1 \le  k \le maxdegree$;
\item \hspace*{20pt} for $1 \le i \le k$ {\bf guess} a full type type $\bar{\alpha}_i$;
\item \hspace*{20pt} if not $locally$-$consistent(\bar{\alpha}, \bar{\alpha}_1, \ldots, \bar{\alpha}_k)$ then {\bf reject};
\item \hspace*{20pt} $level:= level + 1$;
\item \hspace*{20pt} {\bf universally choose} $1 \le i \le k$; let $\bar{\alpha}=\bar{\alpha}_i$;
\item {\bf endwhile}
\item {\bf reject}

\end{itemize}
The function $locally$-$consistent$  checks whether, from a local point of view,
a tree may have a node of full type $\bar{\alpha}$ whose children,
listed from left to right, have full types $\bar{\alpha}_1, \ldots, \bar{\alpha}_k$. 
Namely, it returns {\bf true} if and only if all  of the following conditions hold:

\medskip
\noindent{\em Horizontal conditions:}\\
(h1) $\bar{\alpha}_i(\tprech)=\{\alpha_{i-1}\}$ for $i>1$;  $\bar{\alpha}_1(\tprech)=\emptyset$;\\
(h2) $\bar{\alpha}_i(\tsucch)=\{\alpha_{i+1}\}$ for $i<k$;  $\bar{\alpha}_k(\tsucch)=\emptyset$;\\
(h3) $\bar{\alpha}_i(\tgreath)=\bar{\alpha}_{i-1}(\tprech) \cup \bar{\alpha}_{i-1}(\tgreath)$ for $i>1$; $\bar{\alpha}_1(\tgreath)=\emptyset$;\\
(h4) $\bar{\alpha}_i(\tlessh)=\bar{\alpha}_{i+1}(\tsucch) \cup \bar{\alpha}_{i+1}(\tlessh)$ for $i<k$; $\bar{\alpha}_k(\tlessh)=\emptyset$;\\

\noindent{\em Vertical conditions:}\\
(v1) $\bar{\alpha}(\tsuccv)=\{\alpha_1, \ldots, \alpha_k \}$;\\
(v2) $\bar{\alpha}_i(\tprecv)=\{\alpha \}$ for $1 \le i \le k$;\\
(v3) $\bar{\alpha}(\tlessv)=\bigcup_{1 \le i \le k} (\bar{\alpha}_i(\tsuccv) \cup \bar{\alpha}_i(\tlessv))$;\\
(v4) $\bar{\alpha}_i(\tgreatv)=\bar{\alpha}(\tprecv) \cup \bar{\alpha}(\tgreatv)$ for $1 \le i \le k$;\\

\noindent{\em Free conditions:}\\
(f1) $\bar{\alpha}_i(\tfree)=\bigcup_{j {\not=} i} (\bar{\alpha}_j(\tsuccv) \cup \bar{\alpha}_j(\tlessv)) \cup \bar{\alpha}(\tgreath) \cup \bar{\alpha}(\tprech) \cup \bar{\alpha}(\tsucch) \cup \bar{\alpha}(\tlessh) \cup \bar{\alpha}(\tfree)  $ for $1 \le i \le k$. 

\begin{lemma} \label{l:procedure}
Procedure {\tt {FO$^2$[$\succv, \lessv, \succh, \lessh$]-sat-test}} accepts its input $\phi$ if and only if $\phi$ is satisfied in a finite tree.
\end{lemma}

A matching \ExpSpace-lower bound follows from \cite{Kie02}, where it was shown  that a restricted variant of the two-variable guarded fragment with some unary predicates and
a single binary  predicate that is interpreted as a transitive relation is \ExpSpace-hard. It is not hard to see that the proof presented there works
fine (actually, it is even more natural) if we restrict the class of admissible structures to (finite) trees. Thus we get the following corollary.

\begin{corollary}
Over finite trees the satisfiability problem for each logic between  \GFto{} and \FOtall{} is \ExpSpace-complete.
\end{corollary}

\section{Singular finite trees}

We start this section with establishing the complexity of \FOto{}.

\begin{theorem} \label{t:fotsing}
The satisfiability problem for \FOto{} over finite singular trees is \NExpTime-complete.
\end{theorem}


To show the upper bound we observe that every singularly satisfiable formula has a 
singular model  whose all paths are bounded polynomially.
This fact is a generalisation of Theorem 2.1.1 from \cite{Weis11}, that every \FOt$[<]$ formula $\phi$,
singularly satisfiable over finite words, has a finite singular model with polynomially many elements. Actually, our work
is strongly influenced by the construction from \cite{Weis11}, and, generally, can be seen
as its adaptation to the case of trees. We describe here all the required constructions, but omit some proofs,
as many of them are obtained by obvious adjustments of the corresponding proofs for the case of words. Thus, 
in order to fully understand all the details, we advise the reader to familiarise with 
Chapter 2 of \cite{Weis11}.

Before going further we discuss the main differences with the case of words. The main idea from \cite{Weis11}
is to show that for a given singular word $\str{W} \models \phi$, a letter $a \in \tau_0$ and a given subformula $\xi(x)$ of $\phi$ there
exists a division of $\str{W}$ into polynomially many segments in which, at elements satisfying $a$, the value of $\xi(x)$ is constant.
In our case the role of those segments is played by \emph{slices}, i.e., connected components of trees. We show that each path
intersects polynomially many slices. In \cite{Weis11}  \emph{left} and \emph{right} witnesses are considered. In our case
they correspond to \emph{upper} and \emph{lower} witnesses (which, however, in contrast to the case of words, are not necessarily linearly ordered), but we must also deal with \emph{free} witnesses. 
Finally, the small model
is constructed by picking at most three witnesses for each slice. As the total number of considered slices in a tree may be exponential we have to be 
careful at this point, to avoid choosing too many witnesses from a single path.

Now we turn to technical details. Recall that in the current scenario we have 
four order formulas $x \lessv y$, $x {=} y$, $y \lessv x$ and $x \fw y$. We also  use a shortcut: $x \lesseqv y = x \lessv y \vee x {=} y$. 
The normal form from Lemma \ref{l:normalform} is not very useful since it introduces fresh unary predicates that destroy singularity of models. Thus, we only slightly adjust
formulas by converting them to existential negation form (ENNF).
 A formula $\phi \in \FOto{}$ is in ENNF if it does not contain any universal quantifier, and negations only appear in front of unary predicates or existential quantifiers. Negations in front of order formulas are not allowed.
Obviously, any formula $\phi \in \FOto{}$ is equivalent to a formula in ENNF of size at most $2|\phi|$.

We may view our formulas as positive boolean combinations of order formulas and formulas with at most one free variable. 

\begin{proposition} \label{p:betaform}
 Let $\phi \in \FOto{}$ be a formula in ENNF. Then there exists a number $s \in \Nn$, a positive boolean formula $\beta$ in variables $Z_{\lessv}, Z_=, Z_{\greatv}, Z_{\fw},$ $X_1, \dots, X_s$, and formulas $\phi_1 ,\dots, \phi_s \in \FOto{}$ in ENNF, each with at most one free variable, such that
$\phi = \beta(x \lessv y, x {=} y, y \lessv x, x \fw y, \phi_1 , \dots, \phi_s).$
Moreover 
$\phi \equiv (x \lessv y \wedge \phi_{\upharpoonright x \lessv y}) \vee (x {=} y \wedge \phi_{\upharpoonright x=y}) \vee (y \lessv x  \wedge \phi_{\upharpoonright y \lessv x}) \vee (x \fw y \wedge \phi_{\upharpoonright x \fw y})$
where $\phi_{\upharpoonright x \lessv y} = \beta(\top, \bot, \bot, \bot, \phi_1 , \dots, \phi_s)$, and $\phi_{\upharpoonright \theta}$ is analogously defined for the remaining $\theta$-s.

\end{proposition}

For a finite tree $\tree$ and a set of nodes $P \subseteq T$ we define $max(P)$ as the set of the maximal nodes from $P$ and $min(P)$ as the set of the minimal nodes from $P$, with respect to the order relation $\lessv$. For example $max(T)$ is the set of the leaves and $min(T)$ is the singleton consisting of the root of $\tree$.

\begin{lemma}\label{l:pqr}
 Let $\zeta_1(y), . . . , \zeta_t(y)$ be \FOto{} formulas with $y$ as the only free variable and in ENNF, and let $\tree$ be a finite singular tree.
 Let $\beta$ be a positive boolean formula in the variables $Z_{\lessv}, Z_=, Z_{\greatv}, Z_{\fw}, Y_1 ,\dots , Y_t$, let
$\psi(x, y) = \beta(x \lessv y, x {=} y,$ $y \lessv x, x \fw y,$ $\zeta_1(y), \dots , \zeta_t(y))$,
and let $\phi(x) = \exists y \psi(x, y)$. Let
$P' := \{u \in T  \rvert   \tree \models \psi_{\upharpoonright x \lessv y} [u,v],$ for some $v$ s.t. $u \lessv v \}$,
$Q' := \{u \in T  \rvert   \tree \models \psi_{\upharpoonright y \lessv x}[u,v],$ for some $v$ s.t. $v \lessv u \}$,
$R' := \{u \in T  \rvert   \tree \models \psi_{\upharpoonright x \fw y}[u,v],$ for some $v$ s.t. $u \fw v\}.$
Set $P = max(P')$ and $Q = min(Q' \cup R')$. Then for all $u \in \tree$, $\tree \models \phi[u]$ iff there exists $p \in P$ s.t. $ u \lesseqv p$ or there exists $q \in Q$ s.t. $ q \lesseqv u$ or $\tree \models \psi_{\upharpoonright x {=} y}[u,u]$. 
\end{lemma}


\noindent
{\it Remark}. Notice that on every path in $\tree$ there is at most one point from $P$ and at most one point from $Q$. 

\medskip
Let $a \in \tau_0$ be a letter, $\tree$ a finite singular tree, and $S$ a set of nodes of $\tree$. Then by $S^a$ we denote the set of nodes in $S$ where the letter $a$ occurs. We also say that $S$ is a tree \emph{slice} iff it induces a connected (with respect to the symmetric closure of the child relation $\succv$) subgraph of $\tree$.

\begin{lemma} \label{l:intervals}
Let $\phi \in \FOto{}$ be a formula in ENNF and with one free variable, let $\tree$ be a finite singular tree, and let $a \in \tau_0$. There is a set $S \subseteq T$ which is a union of tree slices in $\tree$ such that: for every $u \in T^a$ we have $\tree \models \phi[u]$ iff $u \in S$; and every path in $\tree$ intersects at most $|\phi|^2$ tree slices from $S$.
\end{lemma}

The proof is inductive: if $\phi=\exists y \psi(x,y)$, for $\phi(x,y)= \beta(x \lessv y, x {=} y, y \lessv x,\linebreak x \fw y, \xi_1(x), \dots, \xi_s(x), \zeta_1(y), \dots, \zeta_t(y))$, then we consider the slices 
obtained inductively for the formulas $\xi_\sigma(x)$. The slices for different $\sigma$-s may overlap. Their endpoints determine a more refined division into slices, such that in each slice, on nodes carrying $a$, the values of all $\xi_\sigma(x)$ are constant.
In each such slice we apply Lemma \ref{l:pqr} to introduce new divisions. Now arguments and calculations similar as in the proof of the  corresponding Lemma 2.1.10 from \cite{Weis11} lead to the desired claim. 


\begin{lemma} \label{l:shortpathssing}
Let $\phi$ be an \FOto{} formula over a signature $\tau$. If $\phi$ is satisfied in a singular tree, then
$\phi$ is also satisfied in a singular tree, in which the length of every path is bounded by $6 \cdot |\tau| \cdot |\phi|^3$.
\end{lemma}

\begin{proof}
We assume that $\phi$ is in ENNF. Let $\tree \models \phi$ be singular, and let $\phi_1, \dots, \phi_k$ be the subformulas of $\phi$ of the form $\exists x \psi$ for some variable $x$ and some formula $\psi$. For every $\kappa \in [1,\dots, k]$ we use Proposition \ref{p:betaform} to find a positive boolean formula such that
$\psi_\kappa(x,y) = \beta(x \lessv y, x{=}y, y \lessv x, x \fw y, \xi_1(x), \dots, \xi_s(x), \zeta_1(y),$ $\dots,$ $\zeta_t(y))$.
For every $a \in \tau_0$ and every $\sigma \in [1, \dots, s]$ let $S_\sigma^a$ be a set as in Lemma \ref{l:intervals} applied to the formula $\xi_\sigma(x)$ and $a$, where every path intersects at most $|\xi_\sigma|^2$ tree slices from $S_\sigma^a$. Thus there is a set $\str{I}^a_\kappa$ of tree slices $\trs$ such that: every path in $\tree$ intersects at most $2 \cdot \sum_{\sigma \in [1,s]}|\xi_\sigma|^2$ of them; $\bigcup_{\trs \in \str{I}^a_\kappa}\trs = \tree$; and there are $\xi_1^\trs, \dots, \xi_s^\trs \in \{\top, \bot\}$ such that $\tree \models \xi_\sigma[u]$ iff $\xi_\sigma^\trs = \top$ for every $u \in \trs$ satisfying $a[u]$. For each $\trs \in \str{I}^a_\kappa$ we consider the formula $\phi^{\trs}_{\kappa} = \exists y \psi_\kappa^\trs(x,y)$, where
$\psi_\kappa^\trs(x,y) = \beta(x \lessv y, x{=}y, y \lessv x, x \fw y,\xi_1^\trs, \dots, \xi_s^\trs, \zeta_1(y), \dots, \zeta_t(y)).$ 
Let
$P'' := \{v \in \tree \ \rvert \ \tree \models \psi^\trs_{\kappa \upharpoonright x \lessv y}[u,v]$  for all $u \lessv v\}$,
$Q'' := \{v \in \tree \ \rvert \ \tree \models \psi^\trs_{\kappa \upharpoonright y \lessv x}[u,v]$  for all $v \lessv u\}$,
$R'' := \{v \in \tree \ \rvert \ \tree \models \psi^\trs_{\kappa \upharpoonright x \fw y}[u,v] \text{ for all $u \fw v$}\}$,
and let $P_\trs = max(P''), Q_\trs = min(Q''), R_\trs = max(R'')$. Let $T_\kappa^a = \bigcup_{\trs \in \str{I}^a_\kappa}(P_\trs \cup Q_\trs \cup R_\trs)$, $T' = \bigcup_{\kappa \in [1,k], a \in \tau_0}T_\kappa^a$, and let $\tree'$ be the restriction of $\tree$ to  $T'$.  Lemma \ref{l:pqr} can be used to prove  that $\tree' \models \phi$. Also it can be shown that the paths of $\tree'$ are bounded as required.
\epr
\end{proof}


\begin{corollary} \label{c:smalltree}
Let $\phi$ be a singularly satisfiable \FOto{} formula. Then it is satisfied in
a singular tree whose number of nodes is exponential in $|\phi|$. 
\end{corollary}
\begin{proof}
 Let $\str{T} \models \phi$ be a singular tree over the signature $\tau$, let $\cT$ be the frame of $\str{T}$. By Lemma \ref{l:shortpathssing} we may assume that all its paths are bounded polynomially. Let $\phi'$ be the normal form formula over signature $\tau'$ from the statement of Lemma \ref{l:normalform}. By that lemma $\phi'$ is satisfiable in a model $\str{T}'$ based on the frame $\cT$. By Lemma \ref{l:narrowtrees} 
we can remove some subtrees from $\str{T}'$ to obtain a model  $\str{T}'' \models \phi'$ with exponentially bounded degree. Again by Lemma \ref{l:normalform}, the restriction of $\str{T}''$ to 
the original signature $\tau$ is a singular model. As its paths are bounded polynomially and the degree of nodes is bounded exponentially, the total number of nodes is bounded exponentially in $|\phi|$ as
required. \epr
\end{proof}
Corollary \ref{c:smalltree} justifies the upper bound from Theorem \ref{t:fotsing}, since for a given $\phi$ we can nondeterministically guess its exponential model and then verify it.

The exponential bound on the degree of nodes in singular models of \FOto{} formulas is essentially optimal. Indeed, let us see that there exists
a formula of size polynomial in $n$ in whose every model the root has $2^n$ children. We use unary predicates $root, elem, b_0, \ldots, b_{n-1}$, and
say that all elements in $elem$ are children of the root:
$\forall x (root(x) \Leftrightarrow \neg \exists y \; y \lessv x) \wedge \forall x (\elem(x) \Leftrightarrow (\neg \exists y (y \lessv x \wedge \neg root(y)))$.
We think that each $v$ in $elem$ encodes a number $0 \le N(v) < 2^n$ such that the $i$-th bit in its binary representation is $1$ iff  
the formula  $\delta_i(x)=\exists y (x \lessv y \wedge b_i(y))$ is satisfied at $v$. In a standard way we can now write a formula $\it{first}(x)$ which says that $N(x)=0$, a formula
$\it{last}(x)$ stating that $N(x)=2^n-1$, and a formula $\it{succ}(x,y)$ saying that $N(y)=N(x)+1$. Now the formula $\exists x \; \it{first}(x) \wedge \forall x (\neg \it{last}(x) \Rightarrow \exists y \; \it{succ}(x,y))$ is as required.
This idea can be easily employed to obtain \NExpTime-lower bound in Theorem \ref{t:fotsing}. 


It turns out, that the ability of speaking about pairs of nodes in free position is crucial for \NExpTime-hardness. Indeed 
if we allow only guarded formulas, we get \PSpace{} complexity. The upper bound in the following theorem can be proved by 
bounding polynomially not only the length of the paths but also the degree of the nodes in models of \GFto{} formulas.

\begin{theorem} \label{t:gfto}
The satisfiability problem for \GFto{} over finite singular trees is \PSpace-complete.
\end{theorem}

Finally we show that augmenting \GFto{} with any of the remaining binary navigational predicates leads to 
\ExpSpace-lower bound over singular trees. 
\begin{theorem} \label{t:lowerbounds}
The satisfiability  problem over singular trees for each of the logics \GFt$[\lessv,\succv]$, \GFt$[\lessv, \succh]$, \GFt$[\lessv, \lessh]$ is
\ExpSpace-hard.
\end{theorem}

\section{Future work}

%

One possible direction of a further research could be investigating the case in which infinite trees are admitted as models. It seems that the complexity results we have obtained
for finite trees can be transfered to this case without major difficulties. 
It could be interesting to examine also the cases in which $\tau_{bin}$ contains $\succv$ but does not contain $\lessv$. A related result
is obtained in \cite{CW13}, where \NExpTime-completeness of \FOt{} with counting quantifiers and arbitrary number of binary symbols, of which fixed two have to be interpreted as child relations in two trees. The trees considered in \cite{CW13} are, however, ranked and unordered.

%

\medskip\noindent
{\bf Acknowledgement.} Similar results were obtained independently in \cite{BBLW13}. The two works were merged into a single paper \cite{BBCKLMW13}. 

\bibliographystyle{plain}
\bibliography{fo2trees_kiero} 

\begin{appendix}

\newtheorem*{lemma-non}{Lemma}
\newtheorem*{theorem-non}{Theorem}

\section{Proof of Lemma \ref{l:procedure}}

\begin{proof}
Assume that $\phi$ is satisfiable. By Lemma \ref{l:shortpaths} and Lemma \ref{l:narrowtrees} there exists a small model $\str{T}\models \phi$.
The procedure accepts $\phi$ by making all its guesses in accordance to $\str{T}$, i.e.~in the first step it sets $\bar{\alpha}$ to be equal
to the full type of the root of $\str{T}$ and then in each step it sets $\bar{\alpha}_i$ to be the full type of the $i$-th child of the previously considered element.
In the opposite direction, from an accepting
(tree-)run $t$  of the procedure we can naturally construct a tree structure $\str{T}_t$, with $1$-types of elements as guessed during the execution. Our procedure
guesses actually not only $1$-types but full types of elements. The function $locally$-$consistent$ guarantees that the full types of elements in $\str{T}_t$ are
indeed as guessed. Since the procedure checks if each of those full types is $\phi$-consistent, then by Proposition \ref{p:consistent_types} we have that $\str{T}_t \models \phi$. 
\epr
\end{proof}

\section{Proof of Lemma \ref{l:pqr}}

\begin{proof}
Suppose $\tree \models \phi[u]$, then there is  $v \in \tree$ such that $\tree \models \psi[u,v]$. If $u \fw v$ then $\tree \models \psi_{\upharpoonright x \fw y}[u,v]$. By definition $u \in R'$ and thus there is $q \in Q$ such that $q \lesseqv u$. The cases $v \lessv u$ and $u \lessv v$ are similar. If $u = v$ then $\tree \models \psi[u,u]$ and thus $\tree \models \psi_{\upharpoonright x = y}[u,u]$.
In the opposite direction, suppose there is $q \in Q$ such that $q \lesseqv u$. Notice that if $q \lesseqv u$ then for every node $v$ we have $q \fw v \Rightarrow u \fw v$. So if $q \in R'$ then there is a node $v$ such that  $\tree \models \psi_{\upharpoonright x \fw y}[u,v]$. Otherwise $q \in Q'$ and there exists a node $v$ such that $\tree \models \psi_{\upharpoonright y \lessv x}[u,v]$. In both cases there is a node $v$ such that $\tree \models \psi[u,v]$ and thus $\tree \models \phi[u]$. The case if there is  $p \in P$ such that $u \lesseqv p$ is similar. If $\tree \models \psi_{\upharpoonright x = y}[u,u]$ then $\tree \models \psi[u,u]$ and thus $\tree \models \phi[u]$.
\epr
\end{proof}

\section{Proof of Lemma \ref{l:intervals}}

\begin{lemma-non} {\bf \ref{l:intervals}}
Let $\phi \in \FOto{}$ over the alphabet $\tau_0$ in ENNF and with one free variable, let $\tree$ be a tree over the alphabet $\tau_0$, and let $a \in \tau_0$. There is a set $S \subseteq T$ which is a union of tree slices in $\tree$ such that for every $i \in \tree^a: \tree \models \phi[u]$ iff $u \in S$; and every path in $\tree$ intersects at most $|\phi|^2$ tree slices from $S$.
\end{lemma-non}

\begin{proof}
 Induction on the structure of $\phi$. We consider only the case when $\phi = \exists y \psi(x,y)$. Otherwise the proof is similar as in the corresponding Lemma 2.1.10 from \cite{Weis11}. Let
\[\psi(x,y) = \beta(x \lessv y, x {=} y, y \lessv x, x \fw y, \xi_1(x), \dots, \xi_s(x), \zeta_1(y), \dots, \zeta_t(y))\]
Applying the inductive hypothesis to the formulas $\xi_\sigma$ , $\sigma \in [1, s]$, let $S_\sigma$ be the set as described in the statement of this lemma, and let $\trs_{(\sigma, k_1)} , \dots , \trs_{(\sigma,k_\sigma )}$ be tree slices such that every path in $\tree$ intersects at most $|\xi_\sigma|^2$ of them and $S_\sigma = \bigcup_{l=1}^{k_\sigma}\trs_{(\sigma,l)}$. We define the set $H = \bigcup_{\sigma=1}^{s}S_\sigma \cup \{r\} \cup \Ll$, where $r$ is the root of $\tree$ and $\Ll$ is the set of leaves in $\tree$.

Looking at each tree slice $\trs$ bounded by points from $H$, the truth values of the formulas $\xi_1, \dots, \xi_s$ remain constant among all points from $\trs^a$. Let $\xi_1^\star, \dots, \xi_s^\star$ be these respective true values. Thus, on all nodes from $\trs^a$, $\phi(x)$ is equivalent to $\exists y \beta(x \lessv y, x{=}y, y \lessv x, x \fw y, \xi_1^\star, \dots, \xi_s^\star, \zeta_1(y), \dots, \zeta_t(y))$. This formula satisfies the requirements of Lemma \ref{l:pqr}, so that the truth of $\phi(x)$ over $\trs^a$ is determined by the relative position of $x$ with respect to $P$, $Q$ and by truth of the formulas $\zeta_1(x), \dots, \zeta_t(x)$ for the nodes in between $P$ and $Q$. We now can construct the set $S$ of all nodes from $\tree^a$ where $\phi(x)$ is true as the union of tree slices bounded by: points from $H$; points that result from applying this lemma to the formulas $\zeta_1(x), \dots, \zeta_t(x)$; or points from $P$ and $Q$ added on every tree slice $\trs$.

We set a path in $\tree$ and count the number of tree slices from $S$ this path intersects. An intersection of a path from $\tree$ with a tree slice is an interval. By the remark made after Lemma \ref{l:pqr} we know there is at most one point from $P$ and $Q$ added on every path in $I$, thus there is at most one point $p \in P$ and $q \in Q$ on every interval. This means we can use the calculations in Lemma 2.1.10 from \cite{Weis11} to achieve at most $|\phi|^2$ intervals on every path in $\tree$.

\epr
\end{proof}

\section{Remaining part of the proof of Lemma \ref{l:shortpathssing}}

We argue that $\tree' \models \phi$.
To see this, we show by induction that for every subformula $\eta$ of $\phi$ with at most one free variable and all $u \in T'$, $\tree \models \eta[u]$ iff $\tree' \models \eta[u]$.

If $\eta$ is an atomic formula or a boolean combination of other formulas then the claim is obvious. Suppose $\eta = \phi_\kappa$ for some $\kappa \in [1,k]$.

Suppose that $u \in \trs$ and $\tree \models \eta[u]$. Then there is a $v \in \tree$ such that $\tree \models \psi_\kappa[u,v]$. Let $\trs \in \str{I}^a_\kappa$ such that $u \in \trs$. We find $\hv \in T'$ such that $\tree \models \psi_\kappa[u,v]$ as follows: if $u \lessv v$ then there is a $\hv \in P_\trs$ such that $v \lesseqv \hv$, if $v \lessv u$ then there is a $\hv \in Q_\trs$ such that $\hv \lesseqv v$, if $u \fw v$ then there is a $\hv \in R_\trs$ such that $v \lesseqv \hv$, if $u = v$ then we set $\hv = u$. Clearly $\tree' \models \psi_\kappa[u, \hv]$ and thus $\tree' \models \eta[u]$.
Suppose now that $u \in T'$ and $\tree' \models \eta[u]$. Then it is easy to see that $\tree \models \eta[u]$.

So because $\tree \models \phi$, we have $\tree' \models \phi$. We now show that paths in $\tree'$ have a bounded length. Set $\trs \in \str{I}^a_\kappa$ and the formula $\psi_\kappa^\trs(x,y)$ as in fragment of this proof from the main part of the paper.
For every $b \in \tau_0$ and every $i \in [1,t]$ let $S^{'b}_i$ be a set as in Lemma \ref{l:intervals} applied to the formula $\zeta_i(y)$, where $S_i^{'b}$ intersects at most $|\zeta_i|^2$ tree slices on every path. Thus there is a set $\str{K}^b_\kappa$ of tree slices $\trs'$ such that every path in $\tree$ intersects at most $2 \cdot \sum_{i \in [1,t]}|\zeta_i|^2$ of them and $\bigcup \str{K}^b_\kappa = \tree$. Set a path in $\tree$ and let $\str{J}^a_\kappa$ be the set of intervals that are the intersections of this path with tree slices from $\str{K}^a_\kappa$. We claim that there is at most one element from $P_\trs^b$ on every $J \in \str{J}^b_\kappa$. Suppose we have $u \in \trs$ and $v_1, v_2 \in J \cap P_\trs^b$. We show that $\tree \models \psi_{\kappa \upharpoonright x \lessv y}^\trs[
u,v_1]$ iff $\tree \models \psi_{\kappa \upharpoonright x \lessv y}^\trs[u,v_2]$. Indeed recall that $\psi_{\kappa \upharpoonright x \lessv y}^\trs(x,y) = \beta(\top, \bot, \bot, \bot,\xi_1^\trs, \dots, \xi_s^\trs, \zeta_1(y), \dots, \zeta_t(y))$. Thus the boolean value of $\beta$ depends on the boolean values of $\zeta_i(y)$. But we assumed that they are the same for $v_1$ and $v_2$. Since $P_\trs$ is a set of maximal nodes then $v_1 = v_2$. We can do analogous calculations for the sets $Q_\trs$ and $R_\trs$.
Altogether the length of every path in $\tree'$ is at most $3 \cdot 2 \cdot \sum_{\kappa \in [1,k], a \in \tau_0, i \in [1,t]}|\zeta_i|^2 \leq 6 \cdot |\tau| \cdot |\phi|^3$.
\epr

\section{Complexity of \GFto{} over singular trees}

In this section we expand our arguments for \PSpace{} upper bound for \GFto{} over singular trees.

A \GFto{} formula $\phi$ is in \emph{normal form} if $\phi=\bigwedge_{i \in I}\forall xy (\eta_i(x,y) \Rightarrow \psi_i(x,y)) \wedge \bigwedge_{i \in J} \forall x (\lambda_i(x) \Rightarrow \exists y (\eta_i(x,y) \wedge \psi_i(x,y)))$, for some disjoint index sets $I$ and $J$, where $\eta_i$ is a guard of the form $x\lessv y$, $y \lessv x$ or $x{=}y$, $\lambda_i(x)$ is an atomic formula $a(x)$ for some unary symbol $a$, and $\psi_i(x,y)$ is a boolean combination of unary atomic formulas. 

We can prove a slightly weaker counterpart of  Lemma \ref{l:normalform} for \GFt. Namely, 
we show that satisfiability of a \GFto{} formula can be reduced to satisfiability of a normal form \GFto{} formula nondeterministically.

\begin{lemma} \label{l:normalformgf}
There exists a nondeterministic procedure {\tt GF$^2[\lessv]$-normalisation}, such that for 
a \GFto{} formula $\phi$ over a signature $\tau$, and a tree frame $\cT$  consisting of at least two nodes the following holds. The formula $\phi$ is  satisfiable
over  $\cT$ (singularly satisfiable over $\cT$)
 if and only if there exists a polynomial execution of {\tt GF$^2[\lessv]$-normalisation} on $\phi$ producing a normal form \GFto{} formula $\phi'$ over a signature $\tau'$ consisting of $\tau$ and some additional unary symbols, satisfiable over $\cT$ (satisfiable over $\cT$ in a model which restricted to $\tau$ is singular). 
\end{lemma}

\begin{proof}
By the work from \cite{ST04} it follows that
for a given \GFt{} formula $\phi$ over a signature $\tau$ there exists a polynomially computable formula $\phi'=\bigwedge_{i \in I} ((\forall x \; r_i(x))\Leftrightarrow \exists x (\lambda_i(x) \wedge \psi_i(x)) \wedge ((\forall x \; r_i(x)) \vee (\forall x \; \neg r_i(x)))) \wedge \bigwedge_{i \in J} \exists x (\lambda_i(x) \wedge \psi_i(x)) \wedge \phi''$, for some disjoint index sets $I$ and $J$,  over a signature consisting of $\tau$ and
some additional unary predicates, where $\lambda_i(x)$ is an atomic formula  $a(x)$ for some unary symbol $a$, $\psi_i(x)$ is a boolean combinations of atoms, $\phi''$ is in normal form, and none of $r_i$-s is used as a guard, such that $\phi$ and $\phi'$ are satisfiable over the same tree frames.
Now for each $i \in I$ we guess whether $\forall x \; r_i(x)$ is satisfied or not and replace the occurrences of $r_i(x)$ and $r_i(y)$ in $\phi'$ 
by $\top$ or $\bot$ appropriately. We thus get a conjunction of a normal form formula, some formulas of the form
$\exists x (\lambda_i(x) \wedge \psi_i(x))$, and some formulas of the form $\neg \exists x (\lambda_i(x) \wedge \psi_i(x))$.
A formula of the last type can be rewritten as $\forall xy (x=y \Rightarrow \neg \lambda_i(x) \vee \neg \psi_i(x))$.
To deal with purely existential statements we introduce a fresh unary predicate $root$ and make it true precisely
at the root of a tree by adding the conjunct $\forall xy (x{=}y \Rightarrow (root(x) \Leftrightarrow \neg \exists y (y \lessv x)))$.
A formula $\exists x (\lambda_i(x) \wedge \psi_i(x))$ can be now rewritten as the normal form conjunct $\forall x (root(x) \Rightarrow \exists y (x \lessv y \wedge (\lambda_i(y) \wedge \psi_i(y)) \vee (\lambda_i(x) \wedge \psi_i(x))))$. This transformation works properly 
over trees containing at least two nodes. 
The describe nondeterministic procedure is thus the required {\tt GF$^2[\lessv]$-normalisation} procedure.
\epr
\end{proof}

Let us see that in an arbitrary (not necessarily singular) model $\str{T}$ of a normal form \GFto{} formula $\phi$ we can find a submodel in which the degree of nodes is bounded polynomially in $|\phi|$ and in the length of the paths of $\str{T}$. As we are able to shorten paths in singular
models to length polynomial in $|\phi|$, this will lead to a polynomial bound on the degree of nodes in singular models of \GFto{} formulas (which, as we have seen, contrasts with the case of \FOto{}).

\begin{lemma} \label{l:narrowertrees}
Let $\phi$ be a normal form \GFto{} formula and let $\str{T} \models \phi$. Then there exists a submodel
$\str{T}' \models \phi$ of $\str{T}$ in which the number of successors of each node is bounded by $max \cdot |\phi|$,
where $max$ is the length of the longest path in $\str{T}$.
\end{lemma}

\begin{proof}
Let $v$ be the root of $\str{T}$. For every conjunct $\phi_i$ of $\phi$ of the form 
$\forall x (\lambda_i(x) \Rightarrow \exists y (\eta_i(x,y) \wedge \psi_i(x,y)))$,
with $\eta_i(x,y)=x \lessv y$ pick a witness $w$ for $v$ and $\phi_i$,
mark $w$ and mark all the elements $u$ such that $\str{T} \models u \lessv w$, i.e.,~the
elements on the path from the root to $w$. Remove all subtrees rooted at successors of $v$
containing no marked elements. 
Repeat this process for all the elements $v$ of $\str{T}$,
say, in the depth-first manner. Note that the structure obtained after each step is a model of
$\phi$, since we explicitly take care of providing lower witnesses, and the upper witnesses are retained
automatically as every element which is not removed from the model is kept together with the whole path from
the root from the original model $\str{T}$.
Let $\str{T}'$ be the structure obtained after the final step of the above procedure.
Observe that the number of marked descendants of an element located at level $l$ is 
bounded by $(l+1) \cdot |\phi|$, thus the degree of each node of $\str{T}'$ is bounded
by $max \cdot |\phi|$ as required. 
\epr
\end{proof}

We recall the statement of Theorem \ref{t:gfto} from the main part of the paper, and prove its
part related to the upper bound. Lower bound is proved in the next section.
\begin{theorem-non} {\bf \ref{t:gfto}.}
The satisfiability problem for \GFto{} over finite singular trees is \PSpace-complete.
\end{theorem-non}
\begin{proof}

We show here that the problem belongs to \PSpace{} by designing an alternating polynomial time procedure. 
We first run the non-deterministic procedure {\tt GF$^2$[$\lessv$]-normalisation} (see Lemma \ref{l:normalformgf}) and obtain a formula $\phi'$ over signature $\tau'$. 
 It remains to test satisfiability of $\phi'$.
  The procedure builds a path in a model  together with the immediate successors of
	its nodes.  
Information about a node $u$ consists of its $1$-type, and a polynomially bounded set of atomic $1$-types
	the \emph{promised types of descendants} of $u$.   The procedure starts from guessing information about the root and
	then moves down the tree in the following way: when inspecting a node $u$ it guesses 
	information about all its children (polynomially many) and then proceeds universally to one of
	them. During the execution the following natural conditions are checked:
	\begin{enumerate}
	\item[(i)] Every guessed atomic type contains precisely one predicate from $\tau$.
	\item[(ii)] The set of promised types of descendants of the current node $u$ is sufficient to provide necessary witnesses for $u$  for conjuncts of $\phi'$ of the form
	$\forall x (\lambda_i(x) \Rightarrow \exists y (x \lessv y \wedge \psi_i(x,y)))$. 
	\item[(iii)] The current node has the required witnesses for the conjuncts of the form
	$\forall x (\lambda_i(x) \Rightarrow \exists y (y \lessv x \wedge \psi_i(x,y)))$ among its ascendants.
	\item[(iv)] The universal part $\forall \forall$ of $\phi'$ is not violated by a pair consisting of the current node $u$ and any of its
	ascendants.
	\item[(v)] Every promised type of a descendant  of the inspected node $u$ is either realised or promised by one of its children.
	\end{enumerate}
	The procedure accepts when it reaches (without violating the above conditions) in at most polynomially many steps a node with
	no promised descendants.
	
	%

The described alternating procedure works in time bounded polynomially in $\phi$, so, as \APTime=\PSpace{} \cite{CKS81},
it can be
also implemented to work in deterministic polynomial space.
We claim that it accepts $\phi$ iff $\phi$ has a finite singular tree model. Assume that $\phi$ is accepted. This means
that $\phi'$ has a tree model which restricted to $\tau$ is singular. By Lemma \ref{l:normalformgf} it follows that
$\phi$ has a singular model. In the opposite direction, let $\str{T} \models \phi$ be singular, and let $\cT$ be the frame of $\str{T}$. By Lemma \ref{l:shortpathssing} we can assume that
the depth of $\str{T}$ is bounded by $6 \cdot |\tau| \cdot |\phi|^3$. By Lemma \ref{l:normalformgf}, {\tt GF$^2$[$\lessv$]-normalisation} can produce $\phi'$ which is satisfiable over $\cT$, say in 
a model $\str{T}'$. By Lemma \ref{l:narrowtrees}, $\phi'$ is also satisfied in a submodel $\str{T}''$ of $\str{T}'$ in which the degree
of every node is bounded by $6 \cdot |\tau| \cdot |\phi|^3 \cdot |\phi'|$. Thus our alternating procedure can make all its guesses in
accordance to $\str{T}''$ and accept.
\epr

\end{proof}


\section{Lower bounds for logics over singular trees}

\begin{theorem}\label{th:singFOto}
  The satisfiability problem for \FOto{} over singular finite trees is
  \NExpTime-hard.
\end{theorem}

\begin{proof}
  We give a~reduction from the satisfiability problem of unary \FOt,
  which is known to be \NExpTime-complete~(see e.g., \cite{EVW02}).
  For a~given $\FOt$ formula $\phi$ over a~unary signature $\tau$ we
  construct an equisatisfiable \FOto{} formula $\transl{\phi}$ over
  the signature $\tau\cup\{\lessv,\elem\}$ where $\elem$ is a~fresh unary
  predicate. Without loss of generality we may assume that $\phi$ is
  built from variables $x,y$, unary predicate symbols, boolean
  connectives $\wedge,\neg$ and existential quantification.

Now we inductively define the
translation $\transl{\phi}$.
\begin{eqnarray*}
\transl{p(x)}&=& \exists y\; x{\lessv}y \wedge p(y)\\
\transl{\neg \phi} &= & \neg\transl{\phi}\\
\transl{\phi_1\wedge\phi_2}& = & \transl{\phi}\wedge\transl{\phi_2}\\
\transl{\exists x\; \psi} &= & \exists x\; \elem(x)\wedge \transl{\psi}
\end{eqnarray*}
Note that $\transl{\phi}$ is a formula of length linear in
($|\phi|$). It remains to be shown that $\phi$ and $\transl{\phi}$ are
equisatisfiable.

For one direction, assume that $\str{A}$ is a~model of
$\phi$. Construct a~tree $\str{T}$ such that all elements of the
universe of~$\str{A}$ are immediate successors of the root of
$\str{T}$ and are labeled $\elem$; each such element $e$ has as many
immediate successors as there are predicates in $\tau$ that are true
of $e$, and each such successor is a~leaf labeled with a~distinct
predicate true of~$e$ in~$\str{A}$, see Figure~\ref{fig:repr}. It can
be easily proved by induction on the structure of $\phi$ that
$\str{T}$ is a~(singular) model of $\transl{\phi}$.
\begin{figure}[htb]
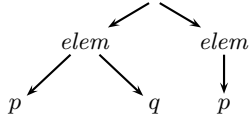

\begin{center}
\[
\begin{array}{c@{\hskip.5cm}c@{\hskip.5cm}c@{\hskip.5cm}c@{\hskip.5cm}c}
&&&\rnode{n00}{} \\[2ex]
&&\rnode{n10}{\elem}&&\rnode{n11}{\elem} \\[3ex]
&\rnode{n1f}{p} & &\rnode{n2a}{q}& \rnode{p}{p}\\[2ex]
\end{array}
\psset{nodesep=2pt}
\ncline{->}{n00}{n10}
\ncline{->}{n00}{n11}
\ncline{->}{n10}{n1f}
\ncline{->}{n10}{n2a}
\ncline{->}{n11}{p}
\]
\caption{Representation of a~structure over the signature
  $\{p,q\}$. There are two elements in the universe; the first belongs
  to the relations $p$ and $q$, the second to $p$.}
\label{fig:repr}
\end{center}
\end{figure}

For the other direction assume that $\str{T}$ is a~model of
$\transl{\phi}$. Construct a~structure $\str{A}$ such that the
universe of $\str{A}$ is the set of nodes labeled $\elem$ in~$\str{T}$
and for all elements $e$ and all predicates $p$, $p(e)$ is true in
$\str{A}$ if and only if there is a~node $e'$ labeled $p$ that is
below $e$ in~$\str{T}$. Again it is easy to prove by structural
induction that $\phi$ is true in~$\str{A}$.  \epr
\end{proof}

\begin{theorem}
  The satisfiability problem for \GFtb{} over singular finite trees is
  \ExpSpace-hard.
\end{theorem}

\begin{proof}
  We give a~reduction from \GFto{} over arbitrary trees. The idea of
  the encoding is the same as in Theorem~\ref{th:singFOto}: a node $e$
  in a~tree is modeled by a~singular node labeled $\elem$ with
  immediate successors encoding predicates true in~$e$. The binary
  predicate $\lessv$ is used to preserve the structure of the tree, the
  additional $\succv$ predicate gives the access to nodes modeling unary
  predicates. In the following reduction, for a~given $\GFto$ formula
  $\phi$ over a~signature $\tau=\tau_0\cup \{\lessv\}$ we construct a
  \GFtb{} formula over the signature $\tau\cup\{\succv,\elem\}$ that is
  satisfiable over singular trees if and only if $\phi$ is satisfiable
  over trees.

  Let us start with a formula ensuring that the underlying structure
  is an encoding of a~tree. The formula $\mathit{tree}$ is defined as
  the conjunction of
\[
\bigwedge_{p\in\tau_0\cup\{\elem\}}\forall x \; p(x) \Rightarrow
\forall y\; y\lessv x\Rightarrow \elem(y)
\]
with 
\[
\forall x\; \elem(x) \Rightarrow \forall y\;
x\lessv y\Rightarrow\bigvee_{p\in\tau_0\cup\{\elem\}} p(y).
\]
It ensures that (unless the tree is trivial, i.e., no node is labeled
at all) each node is labeled with some predicate symbol, all internal nodes
are labeled $\elem$ and only leaves may be labeled with predicates
from $\tau_0$. 

Without loss of generality we may assume that the formula $\phi$ is
built from unary atoms, boolean connectives $\wedge,\neg$ and guarded
existential quantification. The translation $\transl{\phi}$ of
a~formula $\phi$ is defined inductively as follows.

\begin{eqnarray*}
  \transl{p(x)}&=& \exists y\; x\succv y \wedge p(y)\\
  \transl{\neg \phi} &= & \neg\transl{\phi}\\
  \transl{\phi_1\wedge\phi_2}& = & \transl{\phi}\wedge\transl{\phi_2}\\
  \transl{\exists x\; p(x) \wedge\psi(x)} &= 
      & \exists x\; \elem(x)\wedge \transl{p(x)} \wedge\transl{\psi(x)}\\
  \transl{\exists y\; x{\lessv }y \wedge\psi(x,y)} &= 
      & \exists y\;  x{\lessv }y \wedge \elem(y) \wedge\transl{\psi(x,y)}\\
  \transl{\exists y\; y{\lessv }x \wedge\psi(x,y)} &= 
      & \exists y\;  y{\lessv }x \wedge \elem(y) \wedge\transl{\psi(x,y)}
\end{eqnarray*}
Note that $\transl{\phi}$ is a~guarded formula of length linear in
($|\phi|$). Again a~simple inductive argument shows that $\phi$ is
satisfiable if and only if $\mathit{tree}\wedge\transl{\phi}$ has
a~singular tree model. \epr
\end{proof}

\begin{theorem} 
  The satisfiability problems for \GFt$[\lessv, \succh]$ and
  \GFt$[\lessv, \lessh]$ over singular trees are \ExpSpace-hard.
\end{theorem}
\newcommand{\zero}{\mathit{zero}}
\newcommand{\one}{\mathit{one}}
\newcommand{\none}{\mathit{none}}
\newcommand{\Number}{\mathit{number}}
\newcommand{\cell}{\mathit{cell}}
\newcommand{\conf}{\mathit{conf}}

\begin{proof}
  We follow the construction from \cite{Kie02} and give a generic
  reduction from \AExpTime{}. Consider an alternating Turing machine
  $M$ working in exponential time. Without loss of generality we may
  assume that $M$ works in time $2^n$ and that every non-final
  configuration of $M$ has exactly two successor configurations.  Let
  $w$ be an input word of size $n$. Following \cite{Kie02} we
  construct a~formula whose models encode accepting configuration
  trees of machine $M$ on input $w$.
\newcommand{\horizEnc}{
 \begin{array}{c@{\hskip.5cm}c@{\hskip.5cm}c}
\circlenode{p2}{}&\rnode{p3}{\ldots}&\circlenode{p4}{}
\end{array}
\ncline{->}{p2}{p3}
\ncline{->}{p3}{p4}
}
\begin{figure}[htb]
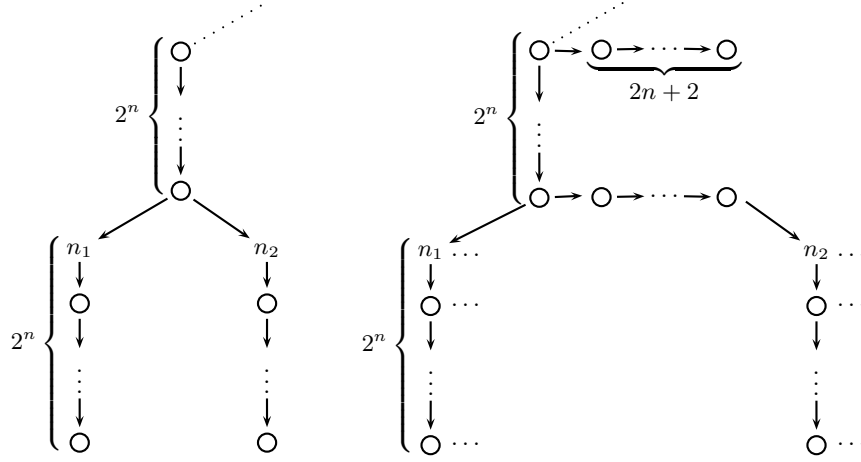

\begin{center}
\[
\begin{array}{c@{\hskip.2cm}c@{\hskip.2cm}c@{\hskip.2cm}c@{\hskip.2cm}c}
&&&\rnode{root}{} \\[2ex]
&&2^n\left\{
\begin{array}{c}
\circlenode{conf0}{} \\[4ex]
\rnode{conf00}{\vdots} \\[2ex]
\circlenode{confn}{} 
\end{array}\right.\mbox{~~~~}
\\[10ex]
&
2^n\left\{
\begin{array}{c}
\rnode{n1f}{n_1} \\[2ex]
\circlenode{conf1}{}\\[4ex]
\rnode{conf1n}{\vdots}\\[2ex]
\circlenode{conf1nn}{}
\end{array}
\right.
&&
\begin{array}{c}
\rnode{n2a}{n_2} \\[2ex]
\circlenode{conf2}{}\\[4ex]
\rnode{conf2n}{\vdots}\\[2ex]
\circlenode{conf2nn}{}
\end{array}
\end{array}
\psset{nodesep=2pt}
\ncline[linestyle=dotted]{root}{conf0}
\ncline{->}{conf0}{conf00}
\ncline{->}{conf00}{confn}
\ncline{->}{confn}{n1f}
\ncline{->}{confn}{n2a}
\ncline{->}{n1f}{conf1}
\ncline{->}{n2a}{conf2}
\ncline{->}{conf1}{conf1n}
\ncline{->}{conf1n}{conf1nn}
\ncline{->}{conf2}{conf2n}
\ncline{->}{conf2n}{conf2nn}
\mbox{~~~~~~~}
\begin{array}{c@{\hskip.02cm}c@{\hskip.2cm}c@{\hskip.2cm}c@{\hskip.2cm}c}
&&\rnode{sroot}{} \\[2ex]
&&2^n\left\{
\begin{array}{c@{\hskip.5cm}c}
\circlenode{sconf0}{} &\rnode{h0}{{\underbrace{\horizEnc}}}\\
& {2n+2}\\
\rnode{sconf00}{\vdots} \\[2ex]
\circlenode{sconfn}{}  &\rnode{hn}{\horizEnc}\rnode{hn2}{}
\end{array}\right.\mbox{~~~~}
\\[10ex]
&
2^n\left\{
\begin{array}{cc}
\rnode{sn1f}{n_1}&\lefteqn{\ldots} \\[2ex]
\circlenode{sconf1}{}&\lefteqn{\ldots}\\[4ex]
\rnode{sconf1n}{\vdots}\\[2ex]
\circlenode{sconf1nn}{}&\lefteqn{\ldots}
\end{array}
\right.
&&
\begin{array}{cc}
\rnode{sn2a}{n_2}&\ldots \\[2ex]
\circlenode{sconf2}{}&\ldots\\[4ex]
\rnode{sconf2n}{\vdots}\\[2ex]
\circlenode{sconf2nn}{}&\ldots
\end{array}
\end{array}
\psset{nodesep=2pt}
\ncline[linestyle=dotted]{sroot}{sconf0}
\ncline{->}{sconf0}{sconf00}
\ncline{->}{sconf00}{sconfn}
\ncline{->}{sconfn}{sn1f}
%
\ncline{->}{sn1f}{sconf1}
\ncline{->}{sn2a}{sconf2}
\ncline{->}{sconf1}{sconf1n}
\ncline{->}{sconf1n}{sconf1nn}
\ncline{->}{sconf2}{sconf2n}
\ncline{->}{sconf2n}{sconf2nn}
\ncline{->}{sconf0}{h0}
\ncline{->}{sconfn}{hn}
\ncline{->}{hn2}{sn2a}
\]
\caption{Left: frame of a~configuration tree in \cite{Kie02}; nodes
  $n_1$ and $n_2$ are siblings. Right: frame of a~configuration tree
  in our encoding; nodes $n_1$ and $n_2$ are not siblings.}
\label{fig:frames}
\end{center}
\end{figure}
In~\cite{Kie02} each configuration is represented by $2^n$ elements of
a~tree, each of which represents a~single cell of the tape of $M$ (see
left part of Figure~\ref{fig:frames}). Each such node is then labeled
with unary predicate symbols from the set
$\{C_1,\ldots,C_n,P_1,\ldots,P_n\}$ to encode the number of
a~configuration (i.e., the depth of the configuration in the
computation tree) and its position (i.e., the number of a~cell) in the
configuration: $C_i(x)$ is true if the $i$-th bit of the configuration
number is 1 and $P_i(x)$ is true if the $i$-th bit of the position
number is 1. Additional predicate symbols are used to encode the
tape symbol and the state of the machine (if it is necessary, i.e., if
the head of of the machine is scanning the cell under
consideration). Here, to encode the numbers, we use additional $2n$
elements that are siblings of the node representing a~cell, see right
part of Figure \ref{fig:frames}. Each of these elements stores 
information about a~single bit using one of two unary predicates
$\zero$ or $\one$. Then the atomic formulas $C_i(x)$ and $P_i(x)$ are
simulated by formulas
\[
\exists y\; x\lessh y \wedge \mathit{Path}_i(y)\wedge \one(y)
\mbox{~~and~respectively~~}\exists y\; x\lessh y \wedge
\mathit{Path}_{n+i}(y)\wedge \one(y)
\]
where the subformula $\mathit{Path}_i(y)$ is defined recursively as
follows.  For the logic \GFt$[\lessv, \succh]$ we define
\begin{eqnarray*}
  \mathit{Path}_0(y)&=& \neg \exists x\; x{\succh} y\\
  \mathit{Path}_{i+1}(y)&=& \exists x\; x{\succh} y \wedge \mathit{Path}_i(x)
\end{eqnarray*}
and for the logic \GFt$[\lessv, \lessh]$ we define
\begin{eqnarray*}
  \mathit{Path}_{\geq 0}(y)&=& \neg \exists x\; x{\lessh} y\\
  \mathit{Path}_{\geq i+1}(y)&=& \exists x\; x{\lessh} y \wedge \mathit{Path}_{\geq i}(x)\\
  \mathit{Path}_{i}(y)&=& \mathit{Path}_{\geq i}(y)\wedge \neg \mathit{Path}_{\geq i+1}(y).
\end{eqnarray*}
Note that in both cases the formula $\mathit{Path}_i$ is guarded and
has polynomial length.  The negated atomic formulas $\neg C_i(x)$ and
$\neg P_i(x)$ are simulated using predicate $\zero$ instead of
$\one$.

Now, having the ability to count, we may encode tape symbols and
states of the machine by simply using more siblings, and we may follow
the lines of the construction in \cite{Kie02} to encode the
computation of $M$.  The only remaining subtle point is that in
\cite{Kie02} the two successor configurations are siblings in
a~computation tree while here they must not be siblings in order not
to mess up the information about numbers --- this may be simply done
by rooting the two configurations at different nodes as shown on
Figure~\ref{fig:frames}.  \epr
\end{proof}

\medskip
\begin{theorem}
The satisfiability problem for \GFto{} over singular trees is \PSpace-hard.
\end{theorem}

\begin{proof}
  We propose a~reduction from the satisfiability of quantified boolean
  formulas, QBF. Let $\psi$ be an instance of QBF problem. Without
  loss of generality we may assume that $\psi$ is of the form
\[
\exists v_k\ldots\exists v_2\forall v_1 \psi'
\]
where the number of all quantifiers ($k$) is even, all even-numbered
variables are existentially quantified, all odd-numbered variables are
universally quantified and $\psi'$ is a~propositional formula over the
variables $v_1,\ldots,v_k$.

We now translate the formula $\psi$ to a~formula  over the signature 
\[\tau=\set{\Root, \leaf,\true, \false, \lessv }
\]
such that $\psi$ is true if and only if
its translation is satisfiable over singular trees.

First, for $i\in \set{0,\dots, k}$ we define auxiliary formulas
$\depth_i$ and $\height_i$.  Let $\depth_0(x)=\Root(x)$ and for $i\geq
1$ let $\depth_i(x)=\exists y \;x{\lessv }y\wedge
\depth_{i-1}(y)$. Intuitively, the formula $\depth_i(x)$ expresses
that the node $x$ occurs at distance at least $i$ from the root. Let
$\height_0(x)=\leaf(x)$, $\height_1(x)=\depth_{k}(x)$ and let
$\height_{i}(x)=depth_{k+1-i}(x)\wedge \neg \depth_{k+2-i}(x)$ for
$i>1$. For $i>0$ the formula $\height_i(x)$ expresses that $x$ is
a~node at depth exactly $k+1-i$; in the construction below, for $i\geq
0$, the formula $\height_i(x)$ will mean that the subtree rooted at $x$ has
height $i$. Note that $\height_i(x)$ is a~guarded formula of length
linear in $i$.

In the following construction a~model of the translation of $\psi$ is
a~tree that describes a~set of valuations justifying that $\psi$ is
true. It is a binary tree of depth $k+1$ where every path describes
a~valuation of variables $v_1,\ldots, v_k$. Every node at height $i$
is labeled either $\true$ or $\false$, which corresponds to a~value of
the variable $v_i$ under a~given valuation. Every non-leaf node at odd
height $i$ has two successors corresponding to the universally
quantified variable $v_{i+1}$; every node at even height $i$ where $i>0$ has one
successor corresponding to the existentially quantified variable $v_{i+1}$. If $k>0$ then let
$\mathit{tree}_k$ be the conjunction of
\begin{eqnarray}
  \exists x\;\Root(x),\\
  \forall x\; \Root(x)&\Rightarrow &(\exists y\;
  x{\lessv }y\wedge\height_k(y)\wedge (\true(y)\vee \false(y)), \\
  \forall x\; \true(x)&\Rightarrow\big( &\height_i(x)\Rightarrow \nonumber\\
&&\big((\exists y\;
  x{\lessv }y\wedge\height_{i-1}(y)\wedge \true(y))\\&\wedge&\;\,(\exists y\;
  x{\lessv }y\wedge\height_{i-1}(y)\wedge \false(y))\big)\; \big)\nonumber\\
  &&\mbox{for all even numbers $2\leq i\leq k$,}\nonumber\\
  \forall x\; \false(x)&\Rightarrow \big(&\height_i(x)\Rightarrow \nonumber\\
&&\big((\exists y\;
  x{\lessv }y\wedge\height_{i-1}(y)\wedge \true(y))\\ &\wedge& \;\,(\exists y\;
  x{\lessv }y\wedge\height_{i-1}(y)\wedge \false(y))\big)\;\big)\nonumber\\
  &&\mbox{for all even numbers $2\leq i\leq k$,}\nonumber\\
  \forall x\; \true(x)&\Rightarrow \big( &\height_i(x)\Rightarrow \nonumber\\
&&\exists y\; x{\lessv }y\wedge\height_{i-1}(y)
  \wedge \big(\true(y)\vee \false(y)\big)\big) \\
  &&\mbox{for all odd numbers $3\leq i<k$,}\nonumber
\\
  \forall x\; \false(x)&\Rightarrow \big( &\height_i(x)\Rightarrow \nonumber\\
&&\exists y\; x{\lessv }y\wedge\height_{i-1}(y)
  \wedge \big(\true(y)\vee \false(y)\big)\big) \\
  &&\mbox{for all odd numbers $3\leq i<k$,}\nonumber\\
  \forall x\; \true(x)&\Rightarrow \big( &\height_1(x)\Rightarrow \nonumber\\
&&(\exists y\;
  x{\lessv }y\wedge \leaf(y))\big),\\
  \forall x\; \false(x)&\Rightarrow \big( &\height_1(x)\Rightarrow \nonumber\\
&&(\exists y\;
  x{\lessv }y\wedge \leaf(y))\big).
\end{eqnarray} 
In the case of $k=0$ the formula $\mathit{tree}_0$  boils down
to $\exists x\;\Root(x)\wedge\forall x\; \Root(x)\Rightarrow (\exists
y\; x{\lessv }y\wedge \leaf(y))$.  Note that $\mathit{tree}_k$ is a~guarded
formula of length polynomial in~$k$.  Now we inductively define the
translation $\transl{\psi'}$ of the quantifier-free formula~$\psi'$.
\begin{eqnarray*}
\transl{\true}&=& \true\\
\transl{\false}&=& \false\\
\transl{v_i}&=& \exists y\; y{\lessv }x \wedge \height_i(y)\wedge \true(y)\\
\transl{\neg \phi} &= & \neg\transl{\phi}\\
\transl{\phi_1\wedge\phi_2}& = & \transl{\phi}\wedge\transl{\phi_2}\\
\transl{\phi_1\vee\phi_2} &= & \transl{\phi}\vee\transl{\phi_2}
\end{eqnarray*}
Note that $\transl{\psi'}$ is a~guarded formula of length polynomial
in ($|\psi'|+k$).  It is not difficult to prove by induction on $k$ (and
by nested structural induction on propositional formulas with free
variables $v_1,\ldots,v_k$) that $\psi$ is true if and only if
$\mathit{tree}_k\wedge \forall x \; \leaf(x)\Rightarrow
\transl{\psi'}$ has a~singular tree model. Each node labeled $\leaf$
in such a~model uniquely determines a~path to a~node labeled $\Root$
and such a~path corresponds to a~valuation of the variables
$v_1,\ldots,v_k$ that makes the formula $\psi'$ true. \epr
\end{proof}


\end{appendix}
\end{document}